\documentclass[12pt, a4paper, notitlepage]{article}
\RequirePackage[cp1251]{inputenc}
\usepackage[english]{babel}
\usepackage{amsthm}
\RequirePackage{amssymb}
\RequirePackage{amsmath} 
\RequirePackage{amscd}
\usepackage[dvips]{graphicx}

\theoremstyle{plain}
\newtheorem{theorem}{Theorem}
\newtheorem{lemma}{Lemma}

\theoremstyle{definition}

\newtheorem*{remark}{Remark}

\begin{document}
\title{Inversion Complexity of Functions\\
of Multi-Valued Logic} 
\author{V.\,V.\,Kochergin\footnote{ Lomonosov Moscow State University
		(Faculty of Mechanics and Mathematics,
		Bogoliubov Institute for Theoretical Problems of Microphysics); vvkoch@yandex.ru}, 
	A.\,V.\,Mikhailovich\footnote{National Research University Higher School of Economics; anna@mikhaylovich.com}}
\date{October, 2015}
\maketitle

\begin{abstract}
The minimum number of NOT gates in a logic circuit computing
a Boolean function is called the inversion complexity of the function.
In 1957, A. A. Markov determined the inversion complexity of every Boolean function
and proved that $\lceil\log_{2}(d(f)+1)\rceil$ NOT gates are necessary and sufficient to
compute any Boolean function $f$ (where $d(f)$ is the maximum number of value changes 
from greater to smaller over all increasing chains of tuples of variables values). 
This result is extended to $k$-valued functions computing in this paper. 
Thereupon one can use monotone functions ``for free'' like in the Boolean case.
It is shown that the minimum sufficient number of non-monotone gates
for the realization of the arbitrary 
$k$-valued logic function $f$
 is equal to $\lceil\log_{2}(d(f)+1)\rceil$
if Post negation (function $x+1 \pmod{k}$) is used  in NOT nodes  and 
is also equal to $\lceil\log_{k}(d(f)+1)\rceil$, 
if {\L}ukasiewicz negation (function $k-1-x$) is used in NOT nodes.
Similar extension for another classical result of A. A. Markov for the inversion complexity of
a system of Boolean functions to $k$-valued logic functions has been obtained.

\medskip

{\it Keywords}: multi-valued logic functions, logic circuits, circuit complexity, nonmonotone complexity, inversion complexity, Markov's theorem.

\end{abstract}

\bigskip

Let $P_k$ be the set of all functions of $k$-valued logic and $M$ be the set of all functions
that are monotone relative to order $0 <1 <\ldots<k-1.$
We will investigate the complexity of the realization of $k$-valued logic functions 
by circuits~\cite{LupMGU} (also known as combinational machine
or circuits of computation~\cite{Sav}) over bases $B$ of the form:
$$
B= M \cup \{\omega_1, \ldots, \omega_p \}, \quad   \omega_i \in P_k \setminus M,~ i=1, \ldots, p, 
$$
where the weight of any function from $M$ equals zero, the weight of function
$\omega_i,$ $i=1,\ldots,p,$ equals 1.

Let us denote the sum of the weights of the elements from circuit $S$
by \emph{non-monotone complexity $I_B(S)$ of circuit $S$ over
basis $B$}. In other words
it is the number of circuit elements corresponding to non-monotone basis funcitons.
Let $f\in P_k,$ $F\subseteq P_k.$ We denote 
the minimum  non-monotone 
complexity of the circuit that realizes function $f$ (system $F$ respectively) by 
\emph{non-monotone complexity $I_B(f)$ of function $f$} (\emph{complexity $I_B(F)$ of the system $F$ respectively}) \emph{over basis $B$}.	

We emphasize two natural bases~--- basis $B_P$
that consists of all non-monotone fuctions and Post negation ($x+1\pmod {k}$),
and basis $B_L$ that consists of all non-monotone functions and 
 {\L}ukasiewicz negation ($k-1-x$). We will use the term \emph{inversion complexity}
 that is similar to the Boolean function case~\cite{Mar57, Mar63}
 because of these two bases, although it is slightly unsuitable.

A.A.~Markov~\cite{Mar57, Mar63} obtained the exact inversion complexity value
for an arbitrary Boolean function or a Boolean function system over basis
$B_0=M\cup \{\overline{x}\}$~\cite{Mar57, Mar63} (the exact statement of this result is given below). 
E.I.~Nechiporuk~\cite{Nech8} obtained the exact inversion complexity value 
for an arbitrary Boolean function realization by a Boolean formula
(this result was reobtained much later in~\cite{Mor08, Mor09}).
Some results dealt with the inversion complexity
can be also found in [9--13].
In this paper classical Markov's results are extended to the case of $k$-valued logic 
functions. The presentation of the results corresponds with
the presentation of Markov's results in~\cite{Jukna}.

\medskip

The set $\{0, 1, \ldots, k-1 \}$ is denoted by $E_k$.
A sequence of tuples 
$$
\tilde{\alpha}_1=(\alpha_{11},\ldots, \alpha_{1n}),~
 \tilde{\alpha}_2=(\alpha_{21},\ldots, \alpha_{2n}),~ \ldots ,~
   \tilde{\alpha}_r=(\alpha_{r1},\ldots, \alpha_{rn})
$$
 from the set $E_k^n$ is called {\it an increasing chain with respect to order
$0<1<\ldots<k-1$} or just {\it chain}, if all tuples
$\tilde{\alpha}_1, \tilde{\alpha}_2, \ldots ,\tilde{\alpha}_r$ are different and
the following inequalities hold
$$
\alpha_{ij} \le \alpha_{i+1,j}, \quad i=1, \ldots, r-1,~~j=1, \ldots, n.
$$
The tuples $\tilde{\alpha}_1$ and $\tilde{\alpha}_r$ are called  {\it initial} and 
{\it terminal} tuples of the chain respectively.

Let  $f(x_1, \ldots,x_n)$ be a function of $k$-valued logic.
An ordered pair of tuples $\tilde{\alpha}=(\alpha_{1},\ldots, \alpha_{n})$ and $\tilde{\beta}=(\beta_{1},\ldots, \beta_{n})$, $\tilde{\alpha}, \tilde{\beta} \in E_k^n,$ is called {\it a jump for the function} $f$,
if

1) ${\alpha_j} \le {\beta_j}$,~~$j=1, \ldots, n$;

2)  $f(\tilde{\alpha}) > f(\tilde{\beta})$.

{\it A jump for a system of functions} is a pair of tuples which is 
a jump for any function of the system.

Let  $F=\{f_1, \ldots , f_m \}$, $m \ge 1$, be a system of 
$k$-valued logic function with arguments $x_1, \ldots x_n$. 
Let $C$ be a chain of the form
$$\tilde{\alpha}_1, \tilde{\alpha}_2, \ldots ,\tilde{\alpha}_r.$$
{\it Decrease $d_C(F)$ of the system $F$ over chain} $C$ is the number
 of jumps for the system $F$ of the form  $(\tilde{\alpha}_i, \tilde{\alpha}_{i+1})$.

{\it Decrease} $d(F)$ of the system $F$ is the maximum
$d_C(F)$ over all chains~$C$.

\vspace{3mm}

Now we can give the exact statement for the Markov's classical result~\cite{Mar57, Mar63}.
Let $F$ be a system of Boolean functions. Then 
$I_{B_0}(F) = \left\lceil \log_{2} (d(F)+1) \right\rceil $.

Let

$$
d(B)=\max \{ d(\omega_1), \ldots, d(\omega_p)  \}.
$$

\bigskip

\begin{theorem} \label{MU}
Let $F$ be a system of $k$-valued logic functions. Then
$$
I_B(F) \ge \left\lceil \log_{d(B)+1} (d(F)+1) \right\rceil .
$$
\end{theorem}

First we prove an auxiliary statement.

\begin{lemma} \label{l_d}
Let $F$ be a system of $k$-valued logic functions. Then
$$
d(F) \le (d(B)+1)^{I_B(F)} -1.
$$
\end{lemma}

\begin{proof}
Let $F=\{f_1, \ldots , f_m \}$, $m \ge 1$, be a set of functions of $k$-valued
logics with arguments $x_1, \ldots x_n$.
The proof is by induction on $I_B(F)$.

If $I_B(F)=0$ the all functions from $F$ are monotone. Hence, $d(F)=0$.

Assume that the assertion is valid for any $G\subset P_k$ such that
$I_B(G) \le I_B(F) -1$. 
Consider circuit $S$ with $n$ inputs $x_1,\ldots,x_n$ which 
realizes function system $F$ and contains exactly $I_B(F)$
elements of unit weight. 
Let us select the first such vertex (according to any correct numeration) and denote
the corresponding gate by $E$. 
Gate $E$ corresponds to $t$-place function $\omega$,
$\omega\in\{\omega_1,\ldots,\omega_p\}$. 
Denote by $h_1(x_1,\ldots,x_n), \ldots, h_t(x_1,\ldots,x_t)$ functions that are
given at the inputs of $E$. 
Denote by $S'$ a circuit that is obtained from the circuit $S$ by 
replacement of gate $E$ with one more input
with variable $y.$
The circuit $S'$ realizes system
$G=\{g_1, \ldots , g_m \}$ with the following properties:
\begin{multline*}
f_i(x_1, \ldots , x_n)=
g_i\left(\omega(h_1(x_1, \ldots , x_n), \ldots, h_t(x_1, \ldots , x_n)), x_1, \ldots , x_n \right), \\
\quad i=1, \ldots, m.
\end{multline*}
Moreover, $I_B(G) \le I_B(F)-1$.

Consider a chain
$$
C=(\tilde{\alpha}_1, \tilde{\alpha}_2, \ldots ,\tilde{\alpha}_{r})
$$
such that $d(F)=d_C(F)$.

Let us consider the sequence $C'$ of $(n+1)$-tuples:
$$
(\omega(h_1(\tilde{\alpha}_1), \ldots, h_t(\tilde{\alpha}_1)), \tilde{\alpha}_1), \ldots ,
(\omega(h_1(\tilde{\alpha}_r), \ldots, h_t(\tilde{\alpha}_r)), \tilde{\alpha}_r).
$$
The sequence $C'$ is not a chain, but it can be split into $p$ 
parts (each part consists of consecutive elements from $C'$) 
$C'_1,$ \ldots, $C'_p$ such that each $C'_j$, $j=1,\ldots,p,$ is a chain and
$p$ satisfies the inequalities $1\le p\le d(B)+1.$

By the induction assumption relation
$$
d_{C_i'}(G) \le d(G) \le  (d(B)+1)^{I_B(G)} -1 = (d(B)+1)^{I_B(F)-1} -1
$$
is valid for all $j$, $j=1,\ldots,p.$
Now, using equalities
$$
f_i(\tilde{\alpha})=
g_i\left(\omega(h_1(\tilde{\alpha}), \ldots, h_t(\tilde{\alpha})), \tilde{\alpha}\right), \quad i=1, \ldots, m,
$$
we get
$$
d_C(F) \le \sum_{i=1}^{p} d_{C_i'}(G) + p-1 \le \sum_{i=1}^{p}( (d(B)+1)^{I_B(F)-1} -1) + p-1
\le (d(B)+1)^{I_B(F)} -1.
$$

Thus, Lemma~1 is proved.
\end{proof}

{\bf Proof of the Theorem~1.}
Lemma~1 implies the inequality
$$
d(F) \le (d(B)+1)^{I_B(F)} -1.
$$
$I_B(F)$ is an integer. Thus, we obtain the necessary estimation.
Thus, Theorem~1 is proved.

 \bigskip
\begin{remark} 
The estimation from Theorem~1 is approximate even if $k=2.$
Indeed, let us consider system $F=\{\overline{x}, \overline y\}$.
The decrease of the system equals 2.
While any circuit, that uses only one non-monotone element, 
realizes a two-argument function with decrease of 1.
Thus, the inversion complexity of the system cannot equal 1 in any basis.
\end{remark}

 \bigskip

Now we pass on to the upper bound estimation.
Let $f_1(x_1,x_2,\ldots,x_n)$,\ldots, $f_s(x_1,x_2,\ldots,x_n)$ be
a tuple of $k$-valued logic functions. 
A function $g(z_1, \ldots, z_s, x_1, x_2, \ldots , x_n)$, such that
$$
\begin{array}{rcl}
g(1,0, \ldots, 0, x_1, x_2, \ldots , x_n) & = & f_1(x_1, x_2, \ldots , x_n),\\
g(0,1, \ldots, 0, x_1, x_2, \ldots , x_n) & = & f_2(x_1, x_2, \ldots , x_n),\\
. & . & .\\
g(0, \ldots, 0, 1, x_1, x_2, \ldots , x_n) & = & f_s(x_1, x_2, \ldots , x_n)\\
\end{array}
$$
is called \emph{$s$-connector} for the tuple
$f_1(x_1,x_2,\ldots,x_n)$,\ldots, $f_s(x_1,x_2,\ldots,x_n)$.

A set of $s$-connectors for a set of $s$-tuples of functions 
(one $s$-connector for each $s$-tuple) is called
\emph{$s$-connector} for the set.

\begin{lemma} 
Let  $B$ be a basis of the form	
$B= M \cup \{\omega(x_1, \ldots, x_q\}$,  $ \omega \in P_k \setminus M$, $q \ge 1$.
Let $F_1 =\{ f_{11}, \ldots f_{s1}\},$ \ldots, $F_M=\{f_{1m}, \ldots, f_{sm} \}$ be 
arbitrary set of $s$-tuples of $k$-valued logic functions.
Then there is an $s$-connector $G$ of the set such that 
$$
I_B(G) \le \max \{I_B(F_1), \ldots, I_B(F_s) \}.
$$
\end{lemma}

\begin{proof} The proof is by induction on $r=\max \{I_B(F_1), \ldots, I_B(F_s) \}$.

If $r=0$ then the functions from $F_i$,  $i=1, \ldots, s,$ are monotone. 
Then let 
$G$ be the following set:
$$
\{ g_j \mid g_j = \max ( \min (\phi(z_1), f_{1j}), \ldots, \min  (\phi(z_s), f_{sj}),~j=1, \ldots, m \},
$$
where
$$
\phi(z)=\left\{\begin{array}{l}
k-1, \mbox{ if } z\neq 0;\\
0, \mbox{ elsewhere.}
\end{array}\right.
$$

 Let $r>0$ (induction step). Denote by $S_i(\tilde{x})$ any circuit with inputs $x_1, x_2,\ldots,x_n$
that realizes the function system $F_i,$ $i=1,\ldots,s,$ which contains
$\max\{I_B(F_i), 1\}$ gates, corresponding to function $\omega$.
Let us select the first vertex (according to any correct numeration) 
corresponding to the function $\omega$ in circuit $S_i(\tilde x)$.
Denote by $h_{i1}(x_1,\ldots,x_n), \ldots, h_{iq}(x_1,\ldots,x_n)$ 
functions that are given at the inputs of the gate. 
 Denote by $S'$ a circuit with inputs $y, x_1, x_2, \ldots, x_n$ which is obtained from the circuit $S$ by replacing the seleted gate with one more input with variable $y.$
 Denote by 
 $f_{ij}'(y, x_1, x_2, \ldots , x_n)$, $j=1, \ldots, m,$ functions that are realized at the outputs of circuit $S_i(y, \tilde{x})$. Then
 \begin{multline*}
 f_{ij}(x_1, x_2, \ldots , x_n)=\\f_{ij}'(
 \omega(h_{i1}(x_1, x_2, \ldots, x_n), \ldots,
 h_{iq}(x_1, x_2, \ldots, x_n)),
 x_1, x_2, \ldots , x_n),\\
 \quad j=1, \ldots, m.
 \end{multline*}

Suppose $F_i'=\{ f_{i1}', \ldots , f_{im}' \}.$
Since $I_B(F_i')\le r-1$, $i=1, \ldots, s$, 
by the induction assumption there is a set of functions
$$
G'= \{ g_j'(z_1, \ldots, z_s, y, x_1, x_2, \ldots , x_n) \mid j=1, \ldots , m \},$$
such that
\begin{gather*}
I_B(G') \le \max \{ I_B(F_1'), \ldots,  I_B(F_s') \} \le r-1; \\
\begin{array}{rcl}
g_j'(1,0, \ldots, 0,y,x_1, x_2, \ldots , x_n) & = & f_{1j}'(y, x_1, x_2, \ldots , x_n), \quad  j=1, \ldots , m; \\
g_j'(0,1, \ldots, 0,y,x_1, x_2, \ldots , x_n) & = & f_{2j}'(y, x_1, x_2, \ldots , x_n), \quad  j=1, \ldots , m; \\
. & . & .\\
g_j'(0,0, \ldots, 1,y,x_1, x_2, \ldots , x_n) & = & f_{sj}'(y, x_1, x_2, \ldots , x_n), \quad  j=1, \ldots , m. \\
\end{array}
\end{gather*}

Let us replace variable $y$ by function
\begin{multline*}
Y(z_1, \ldots, z_s, x_1, x_2, \ldots , x_n)=\\
 \omega   (
  \max ( \min (\phi(z_1),h_{11}(x_1, x_2, \ldots , x_n)), \ldots, 
 \min  (\phi(z_s), h_{s1}(x_1, x_2, \ldots , x_n))),  
\ldots ,  \\
\max ( \min (\phi(z_1),h_{1q}(x_1, x_2, \ldots , x_n)), \ldots, 
 \min  (\phi(z_s), h_{sq}(x_1, x_2, \ldots , x_n)))    )
\end{multline*}
in function $g_j'(z_1, \ldots, z_s, y, x_1, x_2, \ldots , x_n)$, $j=1, \ldots, m$,

Since equalities
$$
\begin{array}{rcl}
Y(1,0, \ldots, 0, x_1, x_2, \ldots , x_n) & = & \omega   (
h_{11}(x_1, x_2, \ldots , x_n), \ldots ,
h_{1q}(x_1, x_2, \ldots , x_n)), \\
Y(0,1, \ldots, 0, x_1, x_2, \ldots , x_n) & = & \omega   (
h_{21}(x_1, x_2, \ldots , x_n), \ldots ,
h_{2q}(x_1, x_2, \ldots , x_n)), \\
. & . & .\\
Y(0,0, \ldots, 1, x_1, x_2, \ldots , x_n) & = & \omega   (
h_{s1}(x_1, x_2, \ldots , x_n), \ldots ,
h_{sq}(x_1, x_2, \ldots , x_n)) \\
\end{array}
$$
are valid, we get that function
\begin{multline*}
g_j(z_1, \ldots, z_s, x_1, x_2, \ldots , x_n)= \\
g_j'(z_1, \ldots, z_s, Y(z_1, \ldots, z_s, x_1, x_2, \ldots , x_n) ,x_1, x_2, \ldots , x_n)
\end{multline*}
is $s$-connector for the tuple $f_{1j}, \ldots, f_{sj}$, $j=1, \ldots, m.$ 
Moreover, there are inequalities $I_B(G) \le 1+I_B(G') \le r$
for the set $G=\{g_1, \ldots , g_m \}$.

Lemma~2 is proved.
\end{proof}

\bigskip

Let $f(x_1,\ldots,x_n)$ be an arbitrary $k$-valued logic function,
$C=(\tilde{\alpha}_1, \tilde{\alpha}_2, \ldots , \tilde{\alpha}_r)$ be
an arbitrary chain of tuples from $E_k^n.$ 
Denote by $u_C(f)$ the maximum length of subsequence 
$\tilde{\beta}_1, \tilde{\beta}_2, \ldots , \tilde{\beta}_t$
of sequence $C$ such that
$f(\tilde{\beta}_1) > f(\tilde{\beta}_2) > \ldots > f(\tilde{\beta}_t)$.

{\it Inversion power $u(f)$ of the function} $f$ is 
the maximum $u_C(f)$ over all chains $C$ from $E_k^n.$
Obvuiously, for any function $f$ the inequalities
$1 \le u(f) \le d(f)+1$ hold.
Moreover, if function $f$ is not monotone then $u(f) \ge 2$.

\emph{Inversion power $u(B)$ of basis $B$}
is the maximum $u(f)$ over all functions $f$ from $B$.

\bigskip

\begin{theorem} 
Let $F$ be a system of $k$-valued logic functions. Then
$$
I_B(F) \le \lceil \log_{u(B)} (d(F)+1) \rceil .
$$
\end{theorem}

\begin{proof}
Let $u(B)=s$.
Suppose $\omega(x_1,\ldots,x_q)\in B$ such that  $u(\omega)=s$. 
Let 
$B'= M \cup \{\omega(x_1, \ldots, x_q) \}$. 
Since $I_{B'}(F) \ge I_B(F)$
it is enough to prove the inequality
$I_{B'}(F) \le \lceil \log_{s} (d(F)+1) \rceil$.
The proof is by induction on
$R(F)= \lceil \log_{s} (d(F)+1) \rceil$.

If $R(F)=0$, then $d(F)=0$. Hence, all the functions from $F$ are monotone.
Thus, $I_B(F)=0$.

For the induction step let $G$ be a set of functions such that $R(G) \le R(F) -1$.
Suppose the Theorem statement is correct for $G$. 

Denote by $T_1$ a set of $n$-tuples of elements from $E_k$ such that 
for any chain $C$ with terminal tuple from $T_1$ the inequality $d_C(F) <s^{R(F)-1}$
holds, that is
$$
T_1 = \{ \tilde{\alpha} \in E_k^n \mid d_C(F) <s^{R(F)-1}~\hbox{\rm for any chain }~C~\hbox{\rm with terminal tuple}~ \tilde{\alpha} \}.
$$

Further, denote by $T_i,$ $i=2, \ldots, s-1$, a set of $n$-tuples with 
elements from $E_k$ such that for any chain of elements from
$E_k^n \setminus (T_1 \cup \ldots \cup T_{i-1})$   
with a terminal tuple from $T_i$ inequality $d_C(F) <s^{R(F)-1}$
holds, that is
\begin{multline*}
T_i = \{ \tilde{\alpha} \in E_k^n  \setminus (T_1 \cup \ldots \cup T_{i-1}) \mid d_C(F) <s^{R(F)-1}~\hbox{\rm for~any~chain}~C, \\ C \subset E_k^n  \setminus (T_1 \cup \ldots \cup T_{i-1}),~\hbox{\rm with~terminal~tuple}~ \tilde{\alpha} \}.
\end{multline*}

Finally, let
$$
T_s =   E_k^n  \setminus (T_1 \cup \ldots \cup T_{s-1}).
$$

Note that if $\tilde{\alpha} \in T_i$ and $\tilde{\beta}  \prec \tilde{\alpha}$ then 
$\tilde{\beta} \in T_1 \cup \ldots \cup T_{i-1}$, $i=1, \ldots, s$.

Now we prove that for any chain $C$ of elements from $T_s$, the
inequality $d_C(F)<s^{R(F)-1}$ also holds.
Assume the converse. Hence, there is a chain $C_s$ with 
initial tuple
$\tilde{\alpha}_s$, $\tilde{\alpha}_s \in T_s,$
such that $d_{C_s}(F) \ge s^{R(F)-1}$. 
Since $\tilde{\alpha}_s\notin T_s,$ there is a chain $C_{s-1}$ with
initial tuple $\tilde{\alpha}_{s-1}$, $\tilde{\alpha}_{s-1} \in T_{s-1}$
and terminal tuple $\tilde{\alpha}_s$, $\tilde{\alpha}_s \in T_s$,
such that $d_{C_{s-1}}(F) \ge s^{R(F)-1}$.
Similarly, for $i=s-2,\ldots, 1,$ there is a chain $C_i$ 
with initial tuple $\tilde{\alpha}_{i}$,  $\tilde{\alpha}_{i} \in T_{i}$, 
and terminal tuple  $\tilde{\alpha}_{i+1}$, $\tilde{\alpha}_{i+1} \in T_{i+1}$, 
such that $d_{C_{i}}(F) \ge s^{R(F)-1}$.

Then for chain  $C=C_1 \cup \ldots  \cup C_s$ the
relations
$$
d_C(F) = d_{C_1}(F) + \ldots + d_{C_s}(F) \ge s \left( s^{R(F)-1}\right) = s^{R(F)} > d(F),
$$
hold. This contradicts the definition of $d(F)$.

Let $f_j\in F= \{ f_1, \ldots , f_m \}.$ Suppose
$$
f_{ij}(x_1, x_2, \ldots, x_n)=
\begin{cases}
0,&\text{if $(x_1, x_2, \ldots, x_n) \in T_1 \cup \ldots \cup T_{i-1}$;}\\
f_{j}(x_1, x_2, \ldots, x_n),&\text{if $(x_1, x_2, \ldots, x_n) \in T_i$;}\\
k-1,&\text{if $(x_1, x_2, \ldots, x_n) \in T_{i+1} \cup \ldots \cup T_{s}$;}
\end{cases}
$$
$i=1, \ldots, s.$

Let 
$$
 F_i = \{ f_{ij} \mid f_j \in F \}, \quad   i=1, \ldots, s.
$$

By the definition of the set $F_i$  the inequalities
$d(F_i) < s^{R(F)-1},$  $ i=1, \ldots , s,$
hold. Hence, inequalities
$$
d(F_i)  \le  s^{R(F)-1}-1, \quad  i=1, \ldots , s,
$$
are valid. Thus,
$$
R(F_i) =  \lceil \log_{s} (d(F_i)+1) \rceil \le \lceil \log s^{R(F)-1} \rceil = R(F) -1,  \quad  i=1, \ldots , s.
$$

By the definition of the value $s=u(\omega)$ there is a chain
$(\beta_{11}, \ldots , \beta_{1q}), (\beta_{21}, \ldots , \beta_{2q}), \ldots , 
(\beta_{s1}, \ldots , \beta_{sq})$,
such that
$\omega(\beta_{11}, \ldots , \beta_{1q}) > \omega(\beta_{21}, \ldots , \beta_{2q}) > \ldots > \omega(\beta_{s1}, \ldots , \beta_{sq})$.

We define functions  $\xi_1, \ldots , \xi_q$ by the following equalities
$$
\xi_j (x_1, \ldots , x_n) = \beta_{ij}, \quad i=1, \ldots, s, \quad j=1, \ldots , q,
$$
which are valid for all tuples  $ (x_1, \ldots , x_n)$ from $T_i$.

Let $b_i =  \omega(\beta_{11}, \ldots , \beta_{1q})$, $i=1, \ldots , s$.

We define functions $\lambda_j (x)$, $j=1, \ldots , k-1$. Let
$$
\lambda_j (x)= 
\begin{cases}
0,&\text{if $x<j$;}\\
1,&\text{if $x \ge j$.}\\
\end{cases}
$$

We define functions $\mu_i (x_1, \ldots , x_n)$, $i=1, \ldots , s$. Let
$$
\mu_i (x_1, \ldots , x_n)= 
\begin{cases}
0,&\text{if $(x_1, \ldots , x_n) \in T_1 \cup \ldots \cup T_{i-1}$;}\\
1,&\text{if $(x_1, \ldots , x_n) \in T_i \cup \ldots \cup T_{s}$.}\\
\end{cases}
$$

Note that all these functions are monotone.

Consider $s$-connector 
$G= \{g_j(z_1, \ldots, z_s, \tilde{x}) \mid j=1, \ldots , m \}$ 
for the tuple of function 
$\{ (f_{1j}(\tilde{x}), \ldots, f_{sj}(\tilde{x})) \mid j=1, \ldots , m \}$.
By Lemma~2 there exists such an $s$-connector.

Replace variable $z_i$, $i=1, \ldots ,s$, by function
$$
Z_i(x_1, \ldots , x_n) = \min \left\{ \lambda_{b_i} \left( \omega (\xi_1 (x_1, \ldots , x_n),
\ldots , \xi_q (x_1, \ldots , x_n) ) \right), \mu_i  (x_1, \ldots , x_n) \right\}.
$$
in function $g_j(z_1, \ldots, z_s, \tilde{x})$, $j=1, \ldots, m$.

Since function $Z_i(x_1, \ldots , x_n)$ 
equals 1 on tuples from $T_i$ and equals 0 on the other tuples, we get that for all tuples $(x_1, \ldots , x_n)$ from $T_i$ inequalities
\begin{multline*}
g_j(Z_1(x_1, \ldots , x_n), \ldots, Z_s(x_1, \ldots , x_n), x_1, \ldots , x_n)=
f_{ij}(x_1, \ldots , x_n) = \\  f_{j}(x_1, \ldots , x_n),
\quad i=1, \ldots , s,~j=1, \ldots, m,
\end{multline*}
are valid.

To realize functions $Z_1, \ldots , Z_s$ 
one have used monotone functions gates and one gate corresponding to function  $\omega$.
By induction assumption we get
$$
I_{B'}(F) \le I_{B'}(G) +1 \le \max \{ I_{B'}(F_1), \ldots  , I_{B'}(F_s) \} + 1 \le
$$
$$
\le \max \{  \lceil \log_s (d(F_1)+1) \rceil, \ldots, \lceil \log_s (d(F_s)+1) \rceil \}
\le  \left\lceil \log_s s^{R(F)-1}\right\rceil +1 = R(F).
$$
That completes induction step.

Theorem~2 is proved.
\end{proof}

\bigskip

If basis $B$ is such that $d(B)+1=u(B)$,
Theorem~1 and Theorem~2 give the exact value for non-monotone complexity 
in basis $B$ for any system of $k$-valued logic functions. 
Obviously, this equality holds for bases $B_P$ and $B_L.$

\bigskip

\begin{theorem}  
Let  $F$ be a system of $k$-valued logic functions. Then
$$
I_{B_P}(F)=\left\lceil \log_{2} (d(F)+1) \right\rceil ,  \qquad
I_{B_L}(F)=\left\lceil \log_{k-1} (d(F)+1) \right\rceil .
$$
\end{theorem}

A Shannon function for inversion complexity over basis $B$ 
 of $n$-argument function 
and a system of $m$ functions are defined in standard way:
$$
I_B(n) = \max_{f \in P_k(n)} I_B(f), \quad
I_B(n,m) = \max_{F=\{f_1, \ldots, f_m \},~ f_j \in P_k(n)} I_B(F).
$$

Let
$$
T(k,n)=(k-1)n- \left\lfloor \frac{(k-1)n}{k} \right\rfloor + 1 =
(k-2)n + \left\lceil \frac{n}{k}   \right\rceil +1.
$$

\smallskip

\begin{theorem} 
For any $n$ and $m$, $n \ge 1$, $m \ge 2$, inequalities
$$
I_{B_P}(n)=   \left\lceil \log_{2} T(k,n) \right\rceil ,  \qquad
I_{B_P}(n,m)=\left\lceil \log_{2}((k-1)n+1)  \right\rceil ;
$$
$$
I_{B_L}(n)=   \left\lceil \log_{k-1}  T(k,n) \right\rceil ,  \qquad
I_{B_L}(n,m)=\left\lceil \log_{k-1}((k-1)n+1)  \right\rceil .
$$
hold.
\end{theorem}

\bigskip

This study (research grant No 14-01-0144) supported by The National Research University~--- Higher School of Economics' Academic Fund Program in 2014/2015.

The first author was supported by the Russian Foundation for Basic Research (project no.~14--01--00598).


\bigskip

\begin{thebibliography}{999}

\bibitem{LupMGU} 
O. B. Lupanov, {\it Asymptotic Estimations of Complexity of Control Systems}, Moscow: Mosc. State Univ. Press (1984) (in Russian).

\bibitem{Sav}
J.\,E. Savage, {\it The complexity of computing}, New York: Wiley  (1976).

\bibitem{Mar57}
 A.\,A. Markov, On the inversion complexity of systems of functions, {\it 
Doklady Academii Nauk SSSR}, {\bf 116}~(6), 917--919 (1957) (in Russian). English translation in: 
{\it J. of ACM}, {\bf 5} (4), 331--334 (1958).

\bibitem{Mar63}
 A.\,A. Markov, On the inversion complexity of systems of Boolean functions, {\it Doklady Academii Nauk SSSR}, {\bf 150} (3), 477--479 (1963) (in Russian). English
translation in: {\it Soviet Math. Doklady}, {\bf 4}, 694--696 (1963).

\bibitem{Hil} 
E.\,N.~Gilbert, Lattice theoretic properties of frontal switching functions
{\it  J. Math. Phys}, {\bf 33}, 56-67 (1954).

\bibitem{Nech8}      
E.\,I.~Nechiporuk, On the complexity of circuits in some bases containing
nontrivial elements with zero weights, {\it Problemy Kibernetiki}, {\bf  8},
123–160 (1962) (in Russian).

\bibitem{Mor08}  
H. Morizumi,
A note on the inversion complexity of Boolean functions
in Boolean formulas,
{\it 
Cornell University Library}, {\tt
arXiv.org > cs >  arXiv:0811.0699.} 


\bibitem{Mor09}  
H. Morizumi, Limiting negations in formulas,
{\it Lect. Notes in Comput. Sci.},
{\bf 5555}, 36th ICALP, Part I,
701--712 (2009).

\bibitem{Fisch_s}
M.\,J. Fischer,
The complexity of negation-limited networks~--- a brief survey, {\it
Springer Lect. Notes in Comput. Sci.}, {\bf 33},  71--82 (1975).

\bibitem{T96}
K. Tanaka, T. Nishino and R. Beals,  Negation-limited circuit complexity of
symmetric functions, {\it Inf. Proc. Lett.}, {\bf 59} (5), 273--279 (1996).

\bibitem{ST}
S. Sung and K. Tanaka, Limiting negations in bounded-depth circuits: an extension of Markovs theorem, {\it Lect. Notes in Comput. Sci.},
{\bf 2906}, 108--116 (2003).


\bibitem{GMOR15}
S. Guo, T. Malkin,
I.\,C. Oliveira and A. Rosen,
The power of negations in cryptography,
{\it Lect. Notes in Comput. Sci.},
{\bf 9014}, 36--65 (2015).


\bibitem{Prep1506}
V.\,V.~Kochergin and
A.\,V.~Mikhailovich,
Some extensions of the inversion complexity
of Boolean functions, {\it 
Cornell University Library}, {\tt
arXiv.org > cs > arXiv:1506.04485.} 

\bibitem{Jukna}        
S. Jukna, {\it Boolean Function Complexity. Advances and Frontiers,} Springer Berlin Heidelberg (2012).


\end{thebibliography}
\end{document}